\documentclass[11pt]{article}
\usepackage{latexsym}
\usepackage{graphicx}
\usepackage{amsmath}
\usepackage{color}
\usepackage{xcolor}
\usepackage{amsfonts}
\usepackage{natbib}
\usepackage{soul}
\usepackage{hyperref}
\usepackage{authblk}
\usepackage{amssymb}
\usepackage{amsthm}
\newtheorem{theorem}{Theorem}

\newtheorem{lemma}[theorem]{Lemma}

\newtheorem{corollary}[theorem]{Corollary}

\theoremstyle{break}
\newtheorem*{remark}{Remark}

\newcommand{\mbf}[1]{\mbox{\boldmath$#1$}}

\title{Quiz Show Games: \\ Searching with Bimodal Hiding}

\author[1]{Mingshi Yao}
\author[2]{Thomas Lidbetter}
\author[1]{Melike Baykal-G\"{u}rsoy}
\affil[1]{Department of Industrial and Systems Engineering, Rutgers University}
\affil[1]{Department of Management Science and Information Systems, Rutgers University}
\date{\today}

\providecommand{\keywords}[1]{\textbf{\textbf{Keywords:}} #1}

\begin{document}
%\date{Today}
\maketitle

\begin{abstract}
\noindent We consider a quiz show game in which a contestant is presented with a sequence of questions. Each time the contestant answers a question correctly, she receives a prize and proceeds to the next question; the probability of answering each question correctly is given. If the contestant answers a question incorrectly, she receives a consolation prize and the game ends. The contestant's problem of determining the optimal order in which to answer questions, for known fixed parameters, is a classic one studied in \cite{kadane1969quiz}, and admits a simple index-based solution. We consider game-theoretic versions of this problem in which a game show host can choose how to allocate a fixed prize budget. Our models are motivated by operational search problems in national security involving reconnaissance missions and inspecting for evidence of nuclear enrichment, as well as certain scheduling problems. We study three variants of the game, corresponding to different ways in which the host can distribute the prize money. For  the first variant, we provide complete closed-form solutions, including equilibrium strategies and the value of the game. We reduce the second variant to a game solved in the literature. For the third variant, we obtain partial results by analyzing a more general game with a geometric structure.
\end{abstract}


\section{Introduction}

In \cite{kadane1969quiz}, the author considered a problem (which we will refer to as the {\em Quiz Show Problem}) in which a contestant must choose which order to answer a finite set of questions. Each question has a {\em main prize} and a {\em consolation prize}, which are both known to the contestant. The probability she answers any given question correctly is also known to the contestant. If she answers a question correctly she receives a certain prize, and can choose which question to answer next (if there are any remaining unanswered questions); if not, she receives a consolation prize and the game ends. The problem of which order to answer the questions to maximize the expected winnings can be found using a simple index-based policy. 

\cite{kadane1969quiz} actually analyzed a wider range of sequential problems, of which the Quiz Show Problem is an example.

The Quiz Show Problem may be viewed as a model for certain sequencing problems relating to national security, where the objective is to gather hidden information. For instance, suppose a military unit wishes to undertake a reconnaissance mission, requiring information gathering in several locations by a drone. At each location, there is a certain value to the information that may be gathered and transmitted back to the base, but there is also some probability that the drone is captured or shot down (although some information may still be gathered). If the drone is not impaired, it continues to the next locations. The problem is to choose what order to visit the locations to maximize the expected value of information gathered. Here, the locations correspond to the questions in the Quiz Show Problem. The probability of capture corresponds to the probability of getting a question wrong, and the value of the information gathered corresponds to the value of the prizes.

Another problem in national security relates to inspecting for evidence of nuclear enrichment. Suppose a body such as the International Atomic Energy Agency (IAEA) is tasked with inspecting locations in a country that is suspected of breaching nuclear non-proliferation agreements. The IAEA wishes to maximize the value of information they gather from their inspections. The amount of information they gather at a particular site depends on whether or not they find evidence of a breach. If they do not find a breach, they may still gather valuable information. There is a certain probability of detecting a breach, depending on the site. If they do find a breach, the host country will expel the inspectors and the process will end, otherwise they may proceed to the next site. The problem is to choose what order to inspect the sites the maximize the expected value of the information gathered. The sites correspond to the questions in the Quiz Show Problem, and the value of the information gathered corresponds to the prizes. Counterintuitively, detecting a breach corresponding to answering a question incorrectly and receiving the consolation prize. In this set-up, the consolation prize could be larger than main prize for each question.

The Quiz Show Problem can also be viewed as a scheduling problem, related to the one introduced by \cite{stadje1995selecting}, and studied independently by \cite{agnetis2009sequencing}. It features a finite set of jobs that must be scheduled to be processed by a machine in some order. Each job is associated with a given reward and a given success probability, which is the probability the machine does not break down when processing the job. A job's reward is collected if it is successfully completed, otherwise no reward is collected and no further jobs may be processed. The problem is to choose which order to process the jobs to maximize the expected reward collected. This can be viewed as a special case of the Quiz Show Problem, where the consolation prizes are all zero. If the consolation prizes are non-zero, the problem has a natural interpretation in the context of scheduling, where we might imagine that the unsuccessful completion of a job results in some partial reward.

In the national security examples we have considered above, it is unrealistic to assume that we know in advance the value of the information in each of the locations. In the scheduling example, the rewards of the jobs may be uncertain. For these reasons, in this paper we consider different variants of the Quiz Show Problem in which the values of the prizes are chosen by an adversary (which may be Nature). In this case, rather than seeking a policy to maximize the expected rewards obtained, we seek robust randomized policies that maximize the worst-case (maximum) reward.

We use a game theoretic setting to analyze these variants, supposing that some fixed prize budget is distributed among the questions by an adversary (the game show host), who wishes to minimize the contestant's winnings. 
We use male pronounces for the host (``h'' for ``host/hider'' and ``he'') and female pronouns for the contestant (``s'' for ``searcher'' and ``she'').
In the following subsection we explain more precisely the game variants, and outline our main results.

\subsection{Problem Definitions and Main Results}

Quiz show games are defined as games between a contestant and a host. There are $n$ questions, $V \equiv \{1,\dots,n\}$ in the quiz, and there are two pots of money totaling $M$ (the main prize money) and $C$ (the consolation prize money). The main prize money is split between the $n$ questions so that the main prize for question $i$ is $q_i M$ and the consolation prize for question $i$ is $q'_i C$, for each $i=1,\ldots,n$, where $\mbf q, \mbf q' \in \mathbb{R}^n_+$ and $\sum_{i=1}^n q_i = \sum_{i=1}^n q'_i = 1$. Depending on the variant of the game we are considering, $\mbf q'$ may be fixed or chosen by the host ($\mbf q$ is always chosen by the host); also, the parameters $M$ and $C$ may be fixed, or they may be chosen by the host, as explained in more detail later in this subsection.

When answering question $i$, the contestant answers correctly with probability $\alpha_i$, where $\alpha_i$ is some parameter with $0 < \alpha_i < 1$, and she receives the main prize money $q_i M$ allocated to that question. She may then answer another question. With probability $1-\alpha_i$, she receives the consolation prize money $q'_i C$ allocated to that question, and can answer no further questions.

We consider three variants of the game.
In each variant, a pure strategy for the contestant is simply a permutation $\sigma: V \rightarrow V$ of the questions, specifying the order in which she chooses to answer them, where $\sigma(i)$ is the $i$th question to be answered. We denote the set of all permutations of $V$ by $\Sigma_n$.

In the first variant of the game, we assume that the consolation prizes for each question are fixed parameters of the game, and the game show host must choose how to distribute the main prize money $M$ among the $n$ questions. In this case case, $M$, $C$ and $\mbf q'$ are fixed parameters of the game, and the game show host only has to choose $\mbf q$. 

In the second variant, we assume that the game show host has a pot of money that can be split in any way he chooses among the main prizes and consolation prizes. In other words, there is a total amount of money $T$, and the game show host can choose any $M$ and $C$ such that $M+C=T$, and any $\mbf q$ and $\mbf q'$, specifying how the prize money is distributed.

In the third and last variant, we assume that $M$ and $C$ are fixed parameters and the game show host must choose how to distribute the main prize money and the consolation prize money. That is, the host chooses only $\mbf q$ and $\mbf q'$. 

We will show in Section~\ref{sec:fixed-q'} that the optimal strategies and value of the game for the first variant can be expressed in closed form. This solution can be viewed as a generalization of the solution of a game in~\cite{lidbetter2020search}.

In Section~\ref{sec:fixed-T} we show that the second variant of the game can be solved algorithmically from previous work, by reducing it to a game studied in~\cite{hellerstein2023game}.

The third variant is the hardest to analyze. In Section~\ref{sec:fixed-C&M} we show that it is a special case of a more general game of a geometric nature, that is itself a generalization of a game studied in~\cite{hellerstein2023game}. We find a way to characterize the set of strategies for the game show host in this more general game, and use this to show that many of his strategies are dominated. We then exhibit a collection of mixed strategies for the host, one of which is always feasible, and give sufficient conditions for its optimality. We use this to give a closed-form solution to the third variant of the quiz show game in the case that the main prize is much larger than the consolation prize. 

\section{Previous work}

Before beginning our analysis of the three game variants, we review some previous work that is of particular relevance to this paper.

\subsection{Quiz Show Problems} \label{sec:kadane}

As mentioned in the Introduction, \cite{kadane1969quiz} proposed a collection of Quiz Show Problems that are strongly related to the topic of this paper. We briefly describe the most relevant of the variants of the Quiz Show Problem. As in our games, a contestant must decide in which order to answer a set $V$ of $n$ questions, where the probability that question $i$ is correctly answered is $\alpha_i$. The main prize awarded for answering question $i$ correctly, $r_i,$ and the consolation prize awarded if it is answered incorrectly, $r'_i,$ are known parameters of the game. In other words, $r_i=q_iM$ and $r'_i=q'_iC$ are fixed. (In fact, in \cite{kadane1969quiz}, rewards are discounted, but we assume no discounting here for simplicity.) A simple interchange argument shows that any policy that maximizes the expected winnings must order the questions in non-increasing order of the index
\begin{align}
y_i \equiv \frac{ \alpha_i r_i + (1-\alpha_i)r'_i }{1-\alpha_i}. \label{eq:index}
\end{align}
This is significant because it means that for our quiz show games, if we fix the strategy of the game show host, the contestant's best response can be easily calculated using this index rule.

\cite{kadane1969quiz} actually goes further, defining a broad category of problems whose solutions can be computed using an index rule such as this.

\subsection{Unreliable Job Scheduling Problem} \label{sec:unreliable}

A special case of the Quiz Show Problem described in Subsection~\ref{sec:kadane} was later studied by~\cite{stadje1995selecting} with discounting, and \cite{agnetis2009sequencing}. In this problem a set $V$ of $n$ jobs must be processed by a machine in some order. When the machine processes a job, there is a probability it breaks down and can process no further jobs. This probability depends on the job, so that  job~$i$ is successfully processed with probability $\alpha_i$, in which case it generates a reward of $r_i$. With probability $1-\alpha_i$ the machine breaks down when it attempts to process job $i$, so that no reward is obtained and the machine can process no further jobs. The problem is to determine an order in which to process the jobs that maximizes the expected reward obtained.

It is clear that the problem is equivalent to the Quiz Show Problem of Subsection~\ref{sec:kadane} with $r_i$ and $r'_i=0$. In fact, \cite{agnetis2009sequencing} consider a more general problem, where the jobs may be processed in parallel by $m$ machines. Further work on this topic can be found in \cite{agnetis2020largest,agnetis2022replication,agnetis2025replication,agnetis2025unreliable}.

For the one machine problem, \cite{agnetis2009sequencing} show that the solution can be viewed as a special case of the problem of maximizing a linear function over a polymatroid. They obtain this insight by showing that the set of feasible policies can be identified with the vertices of the base of an appropriately defined polymatroid, as we now explain.

For any pure strategy $\sigma \in \Sigma_n$ of the contestant (or equivalently, ordering of the jobs in the scheduling problem), we define $\mbf x^\sigma \in \mathbb{R}^n$ by
\[
x^\sigma_i = (1-\alpha_i) \prod_{\sigma^{-1}(j) < \sigma^{-1}(i)} \alpha_j.
\]
The component $x_i^\sigma$ can be interpreted as the probability that the contestant answers all the questions before question $i$ correctly, and answered question $i$ incorrectly.

It was shown in~\cite{agnetis2009sequencing} that the set of all vectors $\mbf x^\sigma$ is the set of vertices of the base of the polymatroid associated with the submodular function $g:2^V \rightarrow \mathbb{R}$ given by
\begin{align}
g(S) = 1- \prod_{i \in S} \alpha_i, \text{ for all } S\subseteq V. \label{eq:f}
\end{align}
Recall that the polymatroid $\mathcal{P}(f)$ associated with a non-decreasing, non-negative submodular function $f$ is defined as 
\[
\mathcal{P}(f) = \{\mbf x \in \mathbb{R}^n_+: \mbf x(S) \le f(S) \text{ for all } S \subseteq V \},
\]
with $x(S) \equiv \sum_{i \in S} x_i$. The base $\mathcal{B}(f)$ of the polymatroid $\mathcal{P}(f)$ is defined as
\[
\mathcal{B}(f) = \{\mbf x \in \mathcal{P}(f): \mbf x(V)=f(V) \}.
\]
The set of mixed strategies for the contestant can therefore be identified with $\mathcal{B}(g)$. 

\subsection{A Polymatroid Game} \label{sec:polymatroid}

In \cite{hellerstein2023game} the authors considered a zero-sum game in which Player 1's mixed strategy set is the base $\mathcal{B}(g)$ of a polymatroid, and Player 2's mixed strategy set is the simplex $\Delta(\mbf w) \equiv \{\sum_{i=1}^n \theta_i w_i \mbf e^i: \sum_{i=1}^n \theta_i =1, \theta_i \ge 0~,i=1,\ldots,n\}$, where $w_1,\ldots,w_n \ge 0$ are fixed and $\mbf e^i$ is the unit vector in the direction $i$. 

For two strategies $\mbf x \in \mathcal{B}(g)$ and $\mbf y \in \Delta(\mbf w)$ of Player 1 and 2, respectively, the payoff is given by the scalar product $\mbf x^T \mbf y$. The maximizing player may be either Player 1 or Player 2, and the solution for each variant is similar.

This game is a special case of the {\em lexicographically optimal base problem}, introduced and solved by \cite{fujishige1980lexicographically}. 

The game of \cite{hellerstein2023game} generalize various search games including the {\em weighted search game} considered in \cite{yolmeh2021weighted} and the {\em search and rescue game} considered in \cite{lidbetter2020search}. We do not describe the solution of the game here, but \cite{hellerstein2023game} show how optimal strategies can be found for both players in strongly polynomial time (in $n$).

The game we introduce in Subsection~\ref{sec:general} generalizes the game of \cite{hellerstein2023game}. In the new model we introduce, Player 2's mixed strategy set is the Minkowski sum of two simplices.

\subsection{Other Relevant Work}

The games considered in this paper can be considered as {\em search games}. Indeed, if the consolation prize for each question is zero, then each of the three variants the quiz show game reduce to the so-called {\em search and rescue game} of \cite{lidbetter2020search} with one hidden target. In this game, a target (say, a lost hiker) is hidden in one of $n$ locations. A searcher searches the locations one-by-one until finding the target. After searching each location, there is a given probability that the searcher becomes incapacitated and cannot continue the search. Considering this as a game against Nature, a pure strategy for Nature is a choice of location at which the target is hidden. A mixed strategy is a probability distribution $q_1,\ldots,q_n$ specifying the probability the target is hidden in each of the $n$ locations. The payoff, which the searcher wants to maximize, is the probability of finding the target.

Considering the hiding locations as quiz show questions and the hiding probabilities as rewards for correctly answering the questions, the probability of finding the target can be interpreted as the expected reward for the quiz show contestant, assuming there are no consolation prizes, and setting the total main prize $M$ money equal to 1. The introduction of consolation prizes expands the idea of finding a single target to a more general setting in which information or prizes are hidden in a bimodal way.

Search games originated from the works of \cite{bram19632} and \cite{Isaacs-Book-1965}, where a time-minimizing searcher attempts to locate a hider. Overviews of search games may be found in \cite{alpern&gal03book}, \cite{garnaev2012search} and \cite{hohzaki2016search}. \cite{lidbetter2025review} describes recent work in cost-minimizing box search games. While most work on search games takes a cost/time-minimizing approach, there are also other search paradigms that focus on the probability of finding the target or targets: see \cite{stone1976theory}, which focuses on one-sided (non-adversarial) search problems.

\section{Strategies and Payoffs}

In this section we discuss how the strategies and payoffs may be expressed in our three variants of the quiz show game.

We have already seen in Subsection~\ref{sec:unreliable} how the contestant's mixed strategy set can be identified with the base of a polymatroid $\mathcal{B}(f)$, where $f$ is given by~(\ref{eq:f}). Suppose a mixed strategy of the contestant chooses each pure strategy $\sigma \in \Sigma_n$ of $V$ with probability $p_\sigma$. Then this strategy is uniquely identified with the vector $\mbf x(\mbf p) \in \mathcal{B}(f)$ given by
\[
\mbf x(\mbf p) \equiv \sum_{\sigma \in \Sigma_n} p_\sigma \mbf x^\sigma.
\]
Conversely, by Carath{\' e}odory’s Theorem, any $\mbf x \in \mathcal{B}(f)$ can be written as a convex combination 
\[
\mbf x = \sum_{\sigma \in \Sigma_n} p'_\sigma \mbf x^\sigma,
\]
of the vertices $\mbf x^\sigma$ of $\mathcal{B}(f)$. This convex combination is not unique in general. The strategy~$\mbf x$ may then be intrepreted as choosing each permutation $\sigma \in \Sigma_n$ with probability $p'_\sigma$.

The problem of expressing an arbitrary point in the base of a polymatroid of dimension $n$ as a convex combination of its vertices can be solved in strongly polynomial time (in $n$) by combining the generic approach of \cite{grotschel2012geometric} with the algorithm of \cite{fonlupt2009strongly} for finding the intersection of a line with a polymatroid. See \cite{hoeksma2014decomposition} for more details.

In some situations, it will be convenient to describe a mixed strategy for the contestant by some $\mbf x \in \mathcal{B}(f)$, and in some situations, we will describe a strategy by explicitly giving probabilities $p_\sigma$ that the contestant chooses each permutation~$\sigma \in \Sigma_n$.

A vector $\mbf x \in \mathcal{B}(f)$ may be given the interpretation that $x_i$ is the probability that all the questions answered before question $i$ were answered correctly, but question $i$ was answered incorrectly. In this case, the probability that all the questions up to and including question $i$ are answered correctly is $\alpha_i x_i/(1-\alpha_i)$. It follows that if the main prize awarded for answering question correctly $i$ is $r_i$ and the consolation prize for answering question $i$ incorrectly is $r'_i$, then the expected payoff against some $\mbf x \in \mathcal{B}(f)$ is
\[
\sum_{i \in V} \frac{\alpha_i}{1-\alpha_i} x_i r_i + x_i r'_i  = \mbf x^T \mbf y,
\]
where $\mbf y \in \mathbb{R}^n$ is given by the index~(\ref{eq:index}) defined in Subsection~\ref{sec:kadane}. Therefore, in each of our three variants of the quiz show game, a (mixed) strategy for the game show host can be defined by some $\mbf y$ satisfying~(\ref{eq:index}), where $\mbf r$ and $\mbf r'$ satisfy differing properties, depending on the variant. 

Recall that in the first variant, the consolation prizes $r'_i$ are fixed, and the host can only choose the main prizes $r_i$, summing to $M$. Thus, he can choose any $\mbf y$ satisfying~(\ref{eq:index}) such that $\sum_{i \in V} r_i = M$. We may interpret the set of all such $\mbf y$ as the set of mixed strategies, where the pure strategies are given by taking $\mbf r = M \mbf e^i$ for each $i \in V$ in~(\ref{eq:index}).

In the second variant, the game show host can split a total amount of money $T$ between all the main prizes and consolation prizes, so that both $r_i$ and $r'_i$ are again decision variables. Thus, he can choose any $\mbf y$ satisfying~(\ref{eq:index}) such that $\sum_{i \in V} (r_i + r'_i) = T$. Again, we may consider the set of all such strategies $\mbf y$ as the mixed strategies, where the pure strategies are given by taking $(\mbf r, \mbf r')=(T \mbf e^i, \mbf 0)$ and $(\mbf r, \mbf r')=(\mbf 0, T \mbf e^i)$ in~(\ref{eq:index}), for each for $i \in V$.

In the third and final variant, the total main prize money is fixed at $M$ and the total consolation prize money is fixed at $C$, so both $r_i$ and $r'_i$ are decision variables for each $i$. In this case, the host can choose any ${y_i=\alpha_i r_i/(1-\alpha_i) + r'_i}$ with $\sum_{i \in V} r_i=M$ and $\sum_{i \in V} r'_i=C$. It is not hard to show that $\mbf y \in \mathbb{R}^n$ satisfies these conditions if and only if it can be expressed as a convex combination of the points $\mbf y^{(i,j)}$ obtained by taking $(\mbf r, \mbf r')=(M \mbf e^i, C \mbf e^j)$ in~(\ref{eq:index}) with $(i,j) \in V^2$. Thus, the set of all $\mbf y^{(i,j)}$ may be considered as the pure strategies in this game.

In all three variants of the game, the expected payoff for strategies $\mbf x$ and $\mbf y$ of the contestant and the host is given by $\bf x^T \mbf y$. Moreover, in each variant, each player has a finite number of pure strategies, so by the minimax theorem for finite zero-sum games, each game has a value and optimal (max-min for the contestant and min-max for the host) mixed strategies. Our aim is to determine these.

\section{Fixed Consolation Prizes} \label{sec:fixed-q'}
In the first variant of the Quiz Show Game, the consolation prizes associated with each question are exogenously given and treated as fixed parameters of the game. The game show host’s decision is therefore limited to allocating the main prize budget \( M \) across the \( n \) questions.

Formally, let \( M,C >0\), and \( \mbf q' \in \mathbb{R}_+^n \) be fixed parameters. The host’s control variable is the allocation probability vector \( \mbf q \in \mathbb{R}_{+}^n \), which determines how the total main prize \( M \) is distributed among the questions. Hence, the strategic problem for the host reduces to choosing \( \mbf q \) subject to the feasibility constraints of the game.

We show in Theorem~\ref{thm:fixed-consolation} that the general form of the optimal strategy for the contestant is to start by answering a subset of the questions with the largest consolation prizes, before randomizing between the remaining questions. The number of questions in that subset depends on the values of the parameters of the game. The intuition is that the contestant does not mind getting a question wrong if its consolation prize is large, so she saves the prizes with small consolation prizes until the end. It is optimal for the host to distribute the main prizes only among the questions with lower consolation prizes, to increase the attraction of these questions to the contestant.

\begin{theorem} \label{thm:fixed-consolation}
    Assume, without loss of generality, that $q'_1<q'_2<\dots<q'_n$.
Let $N= \max \{k \in V:\sum_{j=1}^{k-1}  ( q'_{k} - q'_{j} ) (1-\alpha_j)/\alpha_j \le M/C\}$. Let $\sigma_i$ be the permutation $(n,n-1,\dots,N+1,i,i+1\dots,N,1,\dots,i-1)$. It is optimal for Player 1 to use $\mbf x^*$ which chooses each  $\mbf x^{\sigma_i}$ with probability $p_i^*$, given by
\[
p_i^* =
\frac{(1-\alpha_i)/\alpha_i}
{\sum_{j=1}^{N} (1-\alpha_j)/\alpha_j},~i = 1, \ldots, N.
\]
It is optimal for Player 2 to use the strategy $\mbf y^*$, which allocates all the main prize money to question $i \in V$ with probability $q_i^*$, given by
\[
q_i^* = \begin{cases}
\dfrac{C(1 - \alpha_i)}{M\alpha_i \displaystyle \sum_{j=1}^{N} \frac{1-\alpha_j}{\alpha_j}}
\left(\dfrac{M}{C}
+ \displaystyle \sum_{j=1}^{N} \frac{(1-\alpha_j) q'_j}{\alpha_j}
- q'_i \displaystyle \sum_{j=1}^{N} \frac{1-\alpha_j}{\alpha_j}
\right),
&~i = 1, \ldots, N, \\
 0, & \text{otherwise}.
\end{cases}
\]

The value $v$ of the game is given by
\[
v \equiv
\left(
\prod_{j=N+1}^{n} \alpha_j
\right)
\frac{1 - \prod_{j=1}^{N} \alpha_j}
{\sum_{j=1}^{N} (1-\alpha_j)/\alpha_j}
\left(
M + C \sum_{j=1}^{N} \frac{1-\alpha_j}{\alpha_j} q_j'
\right)
+
C \sum_{j=N+1}^{n} q_j'(1-\alpha_j)\prod_{i=j+1}^{n} \alpha_i.
\]
\end{theorem}
\begin{proof}
First, we show that for any pure strategy $\mbf x^\sigma$ of the contestant, the expected payoff against $\mbf y^*$ is at most $v$. We will explicitly calculate $y^*_i$, for each $i \in V$, in order to determine the best response to $\mbf y^*$.

By definition of $q_i^*$, a straightforward calculation shows that, for every $i \le N$,
\[
y_i^*
= \frac{C}{\sum_{j=1}^N (1 - \alpha_j)/\alpha_j}
\left(
\frac{M}{C} 
+ \sum_{j=1}^N \frac{1 - \alpha_j}{\alpha_j} q_j'
\right).
\]
Let us denote this index by $z$, so that $y^*_i=z$ for all $i \le N$.

For $i > N$, since $q_i^* = 0$, we have $y_i^* = C q_i'$. By the maximality of $N$,
\[
\frac{M}{C} 
+ \sum_{j=1}^N \frac{1 - \alpha_j}{\alpha_j} q_j'
< q_{N+1}' \sum_{j=1}^N \frac{1 - \alpha_j}{\alpha_j},
\]
which implies that for all $i \le N$, we have $y_i^* < C q_{N+1}' = y_{N+1}^*$. Thus, by the monotonicity of $q'_i$,
\[
z=y_1^* = \cdots = y_N^* 
< y_{N+1}^* < \cdots < y_n^*.
\]
It follows that each $\mbf x^{\sigma_i}$ is a best response to $\mbf q^*$. In particular, we can take $\mbf x^ {\sigma_1}$ as a best response. In order to calculate the expected payoff of $\mbf x^{\sigma_1}$ against $\mbf y^*$, we first compute
\begin{align*}
\sum_{j=1}^{N} x_j^{\sigma_1}
&=
\left(
\prod_{j=N+1}^{n} \alpha_j
\right)
\sum_{j=1}^{N}
(1-\alpha_j)\prod_{i=1}^{j-1}\alpha_i \\
&=
\left(
\prod_{j=N+1}^{n} \alpha_j
\right)
\left(
1 - \prod_{j=1}^{N} \alpha_j
\right).
\end{align*}
Also, for $i > N$,
\[
x_i^{\sigma_1}
=
(1-\alpha_i)
\prod_{j=i+1}^{n} \alpha_j.
\]
Thus, the payoff of $\mbf x^{\sigma_1}$ against $\mbf y^*$ is
\begin{align*}
(\mbf x^{\sigma_1})^T \mbf y^*
&=
\frac{C}{\sum_{j=1}^{N} \frac{1-\alpha_j}{\alpha_j}}
\left(
\frac{M}{C}
+
\sum_{j=1}^{N} \frac{1-\alpha_j}{\alpha_j} q_j'
\right)
\left(
\prod_{j=N+1}^{n} \alpha_j
\right)
\left(
1 - \prod_{j=1}^{N} \alpha_j
\right) +
C \sum_{j=N+1}^{n}
q_j' (1-\alpha_j)
\prod_{i=j+1}^{n} \alpha_i \\
&=
\left(
\prod_{j=N+1}^{n} \alpha_j
\right)
\frac{1 - \prod_{j=1}^{N} \alpha_j}
{\sum_{j=1}^{N} (1-\alpha_j)/\alpha_j}
\left(
M + C \sum_{j=1}^{N} \frac{1-\alpha_j}{\alpha_j} q_j'
\right) +
C \sum_{j=N+1}^{n}
q_j' (1-\alpha_j)
\prod_{i=j+1}^{n} \alpha_i \\
&=v.
\end{align*}
Hence, $\mbf y^*$ guarantees an expected payoff of at most $v$ for the host.

Second, we show that for any strategy $\mbf y$ of the host, given by an allocation $\mbf q$ of the main prize, the expected payoff against $\mbf x^*$ is at least $v$.

For $j \le N$, we have
\[
x_j^{\sigma_i}
=
(1-\alpha_j)
\left(
\prod_{k=N+1}^{n}\alpha_k
\right)
\prod_{k=i}^{j-1} \alpha_k,
\]where all the products $\prod_{k=i}^{j-1} \alpha_k$ are taken cyclically on $\{1,\ldots,N\}$ as shown below.
\[
\prod_{k=i}^{j-1} \alpha_k
\equiv
\begin{cases}
\alpha_i \alpha_{i+1} \cdots \alpha_{j-1}, & \text{for}\quad i \le j-1, \\ 
\alpha_i \alpha_{i+1} \cdots \alpha_N \alpha_1 \cdots \alpha_{j-1}, & \text{for}\quad i > j-1.
\end{cases}
\]
Therefore,
\[
\sum_{i=1}^{N} p_i^* x_j^{\sigma_i}=\left(
\prod_{k=N+1}^{n}\alpha_k
\right)
\frac{1-\alpha_j}
{\sum_{k=1}^{N} (1-\alpha_k)/\alpha_k}
\sum_{i=1}^{N} \frac{1-\alpha_i}{\alpha_i}
\prod_{k=i}^{j-1} \alpha_k.
\]
The sum $\sum_{i=1}^{N}
(1-\alpha_i)\prod_{k=i}^{j-1} \alpha_k$
telescopes, and we have 
\[
\sum_{i=1}^{N} p_i^* x_j^{\sigma_i}
=
\left(
\prod_{k=N+1}^{n}\alpha_k
\right)
\frac{1 - \prod_{k=1}^{N} \alpha_k}
{\sum_{k=1}^{N} (1-\alpha_k)/\alpha_k}
\frac{1-\alpha_j}{\alpha_j}.
\]
For $j > N$,
\[
x_j^{\sigma_i}
=
\frac{1-\alpha_j}{\alpha_j}
\prod_{k=j}^{n}\alpha_k
=
(1-\alpha_j)\prod_{k=j+1}^{n}\alpha_k.
\]
Hence,
\begin{align*}
 (\mbf x^*)^T \mbf y
&= \sum_{j=1}^{N} \left( \sum_{i=1}^{N} p_i^* x_j^{\sigma_i} \right)
\left( \frac{M\alpha_j}{1-\alpha_j} q_j + C q_j' \right)
+ \sum_{j=N+1}^{n} (1-\alpha_j) \prod_{k=j+1}^{n} \alpha_k
\left( \frac{M\alpha_j}{1-\alpha_j} q_j + C q_j' \right) \\
&= \left( \prod_{k=N+1}^{n} \alpha_k \right)
\frac{1 - \prod_{k=1}^{N} \alpha_k}{\sum_{k=1}^{N} (1-\alpha_k)/\alpha_k}
\sum_{j=1}^{N} \frac{1-\alpha_j}{\alpha_j}
\left( \frac{M\alpha_j}{1-\alpha_j} q_j + C q_j' \right) \\
&\quad + \sum_{j=N+1}^{n}
\left( M \prod_{k=j}^{n} \alpha_k q_j + C q_j'(1-\alpha_j) \prod_{k=j+1}^{n} \alpha_k \right) \\
&= \left( \prod_{k=N+1}^{n} \alpha_k \right)
\frac{1 - \prod_{k=1}^{N} \alpha_k}{\sum_{k=1}^{N} (1-\alpha_k)/\alpha_k}
\left( M \sum_{j=1}^{N} q_j + C \sum_{j=1}^{N} \frac{1-\alpha_j}{\alpha_j} q_j' \right) \\
&\quad + \sum_{j=N+1}^{n} M \prod_{k=j}^{n} \alpha_k q_j
+ C \sum_{j=N+1}^{n} q_j'(1-\alpha_j) \prod_{k=j+1}^{n} \alpha_k.
\end{align*}
Note that each $q_j$ for $j \le N$ has the same coefficient. We show that this coefficient is smaller than the coefficient of $q_j$ for $j >N$. First note that
\[
\sum_{k=1}^N \frac{1-\alpha_k}{\alpha_k} \ge \sum_{k=1}^N (1-\alpha_k) \prod_{i<k} \alpha_i = 1 - \prod_{k=1}^N \alpha_k.
\]
Also, for $j >N$, we have $\prod_{k=N+1}^n \alpha_k \le \prod_{k=j}^n \alpha_k$, so
\[
\left(
\prod_{k=N+1}^{n}\alpha_k
\right)
\frac{1-\prod_{k=1}^{N}\alpha_k}
{\sum_{k=1}^{N} (1-\alpha_k)/\alpha_k}
\le
\prod_{k=j}^{n}\alpha_k.
\]
Hence, the coefficient of $q_j$ for each $j >N$ is greater than the coefficients of $q_1,\ldots,q_N$, so the expected payoff $(\mbf x^*)^T \mbf y$ is minimized by setting $q_{N+1}=\cdots=q_n=0$. Thus, $(\mbf x^*)^T \mbf y \ge v$. This completes the proof.
\end{proof}

\section{Fixed Total Prize Money} \label{sec:fixed-T}

We now consider the second variant of the game in which the total prize money $T >0$ is a parameter of the game, but the host is free to distribute it among the main prizes and consolations prizes in any way he chooses. %More precisely, a strategy for the host is given by a choice of $M$ and $C$ with $M+C=T$ and a choice of $\mbf q \ge 0$ and $\mbf q' \ge 0$ with $\sum_{i=1}^n q_i = \sum_{i=1}^n q'_i=1$. The interpretation is that the main prize for question $i$ is $q_i M$ and the consolation prize for question $i$ is $q'_i C$. 

Recall that a strategy for the host is a vector $\mbf y \in \mathbb{R}^n$ satisfying~(\ref{eq:index}), such that $\sum_{i=1}^n r_i +r'_i=T$. As usual, $r_i$ and $r'_i$ are the main and consolation prizes for question $i$. For such a strategy of the host and a strategy $\mbf x \in \mathcal{B}(f)$ of the contestant, the payoff is $\mbf x^T \mbf y$. 

Suppose for some question $i$, we have $\alpha_i < 1-\alpha_i$. In this case, for any fixed $\mbf r, \mbf r'$ with $r'_i > 0$, the payoff against any permutation $\sigma$ could be reduced by replacing $r'_i$ by $0$ and $r_i$ by $r_i+r'_i$. Hence, any strategy with $r'_i >0$ is dominated. Similarly, if $\alpha_i > 1-\alpha_i$, any strategy with $r_i >0$ is dominated. If $\alpha_i=1-\alpha_i$, then any strategy for which $r_i+r'_i$ is equal to some constant has the same expected payoff. In other words, if the contestant is more likely to answer a question correctly than incorrectly, the host should set the consolation prize for that question to zero, and if the contestant is more likely to answer a question incorrectly, the host should set the main prize for the question to zero.

We summarize these observations in the following lemma.

\begin{lemma} \label{lem:fixed-T}
    In the variant of the game with fixed total prize money, a strategy $\mbf y$ for the host given by $y_i = (\alpha_i r_i + (1-\alpha_i)r'_i )/(1-\alpha_i)$ is undominated if and only if $r'_i=0$ for all $i$ with $\alpha_i < 1/2$ and $r_i=0$ for all $i$ with $\alpha_i > 1/2$. If $\alpha_i=1/2$, then all strategies for which $r_i+r'_i$ is equal to some constant have the same expected payoff.
\end{lemma}
Lemma~\ref{lem:fixed-T} shows that we may assume, without loss of generality, that the host sets $r'_i=0$ for $\alpha_i < 1/2$ and $r_i=0$ for all $\alpha_i \ge 1/2$. Thus, we may recharacterize a strategy for the host as some $\mbf \theta \in \mathbb{R}^n_+$ with $\sum_{i=1}^n \theta_i=1$, with the interpretation that if $\alpha_i < 1/2$, then $r_i=\theta_i T$ and $r'_i=0$ and if $\alpha_i \ge 1/2$, then  $r_i=0$ and $r'_i=\theta_i T$. Let
\[
w_i = \begin{cases}
    \frac{  \alpha_i}{1-\alpha_i} T & \text{ if } \alpha_i <1/2 \text{ and} \\
    T & \text{ if } \alpha_i \ge 1/2,
\end{cases}
\]
and identify a strategy $\mbf \theta$ for the host with the vector $\mbf y = \sum_{i \in V} \theta_i w_i \mbf e^i  \in \Delta(\mbf w)$. 
Then the expected payoff of $\mbf y$ against a contestant strategy $\mbf x \in \mathcal{B}(f)$ is equal to $\mbf x^T \mbf y$.
This is a special case of the game studied in~\cite{hellerstein2023game}, as described in Subsection~\ref{sec:polymatroid}, and the solution follows from the results of that paper. We refer the reader to the aforementioned paper for the precise solution, as it cannot be concisely expressed here in closed form.

\section{Fixed Main and Consolation Prize Money} \label{sec:fixed-C&M}

Recall that in this final variant, the total amount of main prize money and the total amount of consolation prize money are given by fixed parameters $M,C >0$. The host's pure strategies are given by the points $\mbf y^{(i,j)}$ obtained by taking $(\mbf r, \mbf r')=(M \mbf e^i, C \mbf e^j)$ in~(\ref{eq:index}), with $(i,j) \in V^2$. The mixed strategies are the set of all ${y_i=\alpha_i r_i/(1-\alpha_i) + r'_i}$ with $\sum_{i \in V} r_i=M$ and $\sum_{i \in V} r'_i=C$.

As $C \rightarrow 0$, this variant of the game reduces to a special case of the first variant, discussed in Section~\ref{sec:fixed-q'}.

As $M \rightarrow 0$, the game becomes an example of the more general game studied in \cite{hellerstein2023game} (see Subsection~\ref{sec:polymatroid}).

For arbitrary $M$ and $C$, we analyze this game by considering a more general model, presented in the next subsection.

\subsection{General Model Definition} \label{sec:general}

We consider a two player zero-sum game between a maximizing Player~1 and a minimizing Player~2. In order to define the players' strategy sets, we first define a submodular function $f:2^V \rightarrow \mathbb{R}_+$, with $V=\{1,\ldots,n\}$. We assume that $f$ is non-negative and non-decreasing (with respect to set inclusion), and that $f(\emptyset)=0$. We also define non-negative vectors~${\mbf a,\mbf b \in \mathbb{R}^n_+}$.

Any permutation $\sigma$ of $V$ corresponds to a vertex $\mbf x^\sigma$ of the base polyhedron $\mathcal{B}(f)$, and these vertices are the pure strategies of Player 1. Hence, the set of mixed strategies of Player 1 is simply~$\mathcal{B}(f)$.

For Player 2, a pure strategy is a pair $(i,j) \in V^2$. For a pure strategy $\mbf x^\sigma$ of Player 1 and a pure strategy $(i,j)$ of Player 2, the payoff is defined as 
\[
P(\mbf x^\sigma, (i,j)) = a_i x^{\sigma}_i + b_j x^{\sigma}_j.
\]
By linearity, the payoff of a mixed strategy $\mbf x$ of Player 1 against $(i,j)$ is given by a similar expression (with $\mbf x$ replacing $\mbf x^\sigma$). We may rewrite the payoff $P(\mbf x,(i,j))$ as $P(\mbf x,(i,j)) = \mbf x^T \mbf y$, where $\mbf y$ is defined by
\[
\mbf{y} \equiv a_i \mbf{e^i} + b_j \mbf{e^j}.
\]
%Let $\Delta(\mbf a)$ denote the simplex given by all convex combinations of vectors $a_i \mbf e^i,~i\in V$, and similarly for $\Delta(\mbf b)$. 
It follows from standard linear algebra that the convex hull of all pure Player 2 strategies $a_i \mbf e^i + b_j \mbf e^j$ is equal to the Minkowski sum $\Delta(\mbf a)+\Delta(\mbf b)$. This is Player 2's mixed strategy set.

The payoff for a Player 1 mixed strategy $x \in \mathcal{B}(f)$ and a Player 2 mixed strategy $y \in \Delta(\mbf a)+\Delta(\mbf b)$ is given by
\[
P(\mbf x, \mbf y )= \mbf{x}^T \mbf{y}.
\]

The relation to the third variant of the Quiz Show Game should be clear. If we set $a_i = M\alpha_i/(1-\alpha_i)$ and $b_i = C$ for all $i \in V$, then $\Delta(\mbf a) + \Delta(\mbf b)$ is the set of feasible strategies for the game show host. For a strategy $\mbf x \in \mathcal{B}(g)$ of the contestant (where $g$ is given by~(\ref{eq:f})) and a strategy $\mbf y \in \Delta(\mbf a) + \Delta(\mbf b)$ of the host, the payoff of the game is $\mbf x^T \mbf y$.

\subsection{Results}

We make the standing assumption for the rest of this section that $\mbf a$ and $\mbf b$ satisfy $a_1 \le \cdots \le a_n$ and $b_1 \ge \cdots \ge b_n$. Note that this is a valid assumption for the quiz show game, without loss of generality.

We first characterize Player 2's mixed strategy set by a set of linear inequalities.

\begin{lemma} \label{lem:ineq}
    The vector $\mbf y\in \mathbb{R}^n_+$ lies in Player 2's mixed strategy set $ \Delta(\mbf a) + \Delta(\mbf b)$ if and only if $\mbf y$  satisfies the following inequalities.
    \begin{align}
        a_m \sum_{i=1}^m \frac{y_i}{a_i} + b_m \sum_{i=m+1}^n \frac{y_i}{b_i} &\ge a_m + b_m \text{ for all }m \in V; \label{eq:ineq1}\\
        b_m \sum_{i=1}^m \frac{y_i}{b_i} + a_m \sum_{i=m+1}^n \frac{y_i}{a_i} &\le a_m + b_m \text{ for all }m \in V. \label{eq:ineq2}
    \end{align}
\end{lemma}

\begin{proof}
We first assume that $ \Delta(\mbf a) + \Delta(\mbf b)$, and we show that $\mbf y$ satisfies~(\ref{eq:ineq1}) and~(\ref{eq:ineq2}). In this case, we may write
\[
y_i = a_i q_i + b_i q_i', 
\qquad i = 1,\ldots,n,
\]
with each $q_i$ and $q_i'$ satisfying
\[
q_i, q_i' \ge 0,
\qquad
\sum_{i=1}^{n} q_i = 1,
\qquad
\sum_{i=1}^{n} q_i' = 1.
\]
We first prove~(\ref{eq:ineq1}). For $i \le m$, we use our expression for $y_i$ to write
\begin{align}
a_m \frac{y_i}{a_i}
=
a_m q_i + a_m \frac{b_i}{a_i} q_i'
&\ge
a_m q_i + b_m q_i',
\label{eq:small-i}
\end{align}
where the inequality follows from the fact $a_i$ is non-increasing and $b_i$ is non-decreasing in $i$.

Similarly, for $i > m$, we have
\begin{align}
b_m \frac{y_i}{b_i}
=
b_m \frac{a_i}{b_i} q_i + b_m q_i'
&\ge
a_m q_i + b_m q_i'. \label{eq:big-i}
\end{align}
Summing~(\ref{eq:small-i}) over $i \le m$ and~(\ref{eq:big-i}) over $i > m$, we obtain
\[
a_m \sum_{i=1}^{m} \frac{y_i}{a_i}
+
b_m \sum_{i=m+1}^{n} \frac{y_i}{b_i}
\ge
a_m \sum_{i=1}^{n} q_i
+
b_m \sum_{i=1}^{n} q_i'
=
a_m + b_m,
\]
which is exactly~(\ref{eq:ineq1}). 

We prove~(\ref{eq:ineq2}) similarly, using the monotonicity of $a_i$ and $b_i$ to write
\begin{align*}
    b_m \sum_{i=1}^{m} \frac{y_i}{b_i}
+
a_m \sum_{i=m+1}^{n} \frac{y_i}{a_i} &= b_m \sum_{i=1}^{m} \left( \frac{a_i}{b_i} q_i + q'_i \right)
+
a_m \sum_{i=m+1}^{n} \left( q_i + \frac{b_i}{a_i} q'_i \right) \\
& \le \sum_{i=1}^{m} \left( a_m q_i + b_m q'_i \right)
+
\sum_{i=m+1}^{n} \left( a_m q_i +  b_m q'_i \right) \\
&= a_m+ b_m,
\end{align*}
which is~(\ref{eq:ineq2}). So every $\mbf y \in  \Delta(\mbf a) + \Delta(\mbf b)$ satisfies~(\ref{eq:ineq1}) and ~(\ref{eq:ineq2}).

Now assume that $\mbf y \ge \mbf 0$ satisfies~(\ref{eq:ineq1}) and~(\ref{eq:ineq2}) for every $m$.
We show that $\mbf y \in \Delta(\mbf a) + \Delta(\mbf b)$. In particular, we will construct $\mbf q, \mbf q' \in \mathbb{R}^n_+$ with $\sum_{i=1}^n q_i = \sum_{i=1}^n q'_i = 1$ such that $y_i = a_i q_i + b_i q'_i$ for each $i \in V$.

We claim that it is sufficient to find  a $\mbf q \in \mathbb{R}^n$ such that
\begin{align}
0 \le q_i &\le \frac{y_i}{a_i} \text{ for } i \in V ,
\label{eq:q} \\
\sum_{i=1}^{n} q_i &= 1, \label{eq:q-sum} \text{ and} \\
T \equiv \sum_{i=1}^{n} \frac{y_i - a_i q_i}{b_i} 
&=
  1.\label{eq:q'}
\end{align}
Indeed, given such a $\mbf q$, we may define $\mbf q'$ by
\[
q_i' := \frac{y_i - a_i q_i}{b_i}.
\]
The inequalities~(\ref{eq:q}) ensure that each $q'_i \ge 0$, and~(\ref{eq:q'}) ensures that $\sum_{i=1}^n q'_i = 1$. Furthermore, it is clear from the definition of $\mbf q'$ that $y_i=a_i q_i + b_i q'_i$ for each $i \in V$.

So it remains to prove the existence of $\mbf q$ satisfying (\ref{eq:q})-(\ref{eq:q'}).

Consider the set
\[
U :=
\left\{
\mbf q \in \mathbb{R}^n :
0 \le q_i \le \frac{y_i}{a_i},\;
\sum_{i=1}^{n} q_i = 1
\right\},
\]
and define $L:U \rightarrow \mathbb{R}$ by
\[
L( \mbf q) = \sum_{i=1}^{n} \frac{a_i}{b_i} q_i.
\]
We will show that there exists some $\mbf q \in U$ such that $L(\mbf q) = T$. Clearly, such a~$\mbf q$ satisfies inequalities~(\ref{eq:q}) to~(\ref{eq:q'}).

Because~(\ref{eq:ineq1}) with $m=n$ gives
\[
a_n \sum_{i=1}^{n} \frac{y_i}{a_i} \ge a_n + b_n,
\]
we have
\begin{align}
\sum_{i=1}^{n} \frac{y_i}{a_i} &\ge 1 + \frac{b_n}{a_n} > 1. \label{eq:ell}
\end{align}

We are going to compute $\arg \min L(\mbf q)$ and $\arg \max L(\mbf q)$. Let $\ell \in V$ be such that
\[
\sum_{i=1}^{\ell-1} \frac{y_i}{a_i} < 1 \le
\sum_{i=1}^{\ell} \frac{y_i}{a_i}.
\]
By~(\ref{eq:ell}), such an $\ell$ exists and is uniquely defined for a given $\mbf y$.

Let $\mbf q^{(1)}$ be given by
\[
q^{(1)}_i =
\begin{cases}
\dfrac{y_i}{a_i}, & i < \ell, \\[6pt]
1 - \displaystyle\sum_{i=1}^{\ell-1} \frac{y_i}{a_i}, & i = \ell, \\[8pt]
0, & i > \ell.
\end{cases}
\]

It is clear that $\mbf q^{(1)} \in U$. Let $\mbf q$ be some other element of $U$, and consider
$L(\mbf q) - L(\mbf q^{(1)})$. Because every $\mbf q \in U$ satisfies $0 \le q_i \le \frac{y_i}{a_i}$,
it follows that
\[
q_i - q_i^{(1)} \le 0 \quad (i < \ell),
\qquad
q_i - q_i^{(1)} \ge 0 \quad (i > \ell).
\]
Also, by the monotonicty of $\mbf a$ and $\mbf b$, we have
\[
\frac{a_i}{b_i} \le \frac{a_\ell}{b_\ell} \quad \text{for } i < \ell
\qquad \text{and} \qquad
\frac{a_i}{b_i} \ge \frac{a_\ell}{b_\ell} \quad \text{for } i > \ell.
\]
It follows from these four sets of inequalities that $(a_i/b_i)(q_i-q^1_i) \ge (a_\ell/b_\ell)(q_\ell-q^1_\ell)$ for $i < \ell$ and $i > \ell$. Hence,
\[
L(\mbf q) - L(\mbf q^{(1)})
\ge
\frac{a_\ell}{b_\ell} \sum_{i=1}^{n} (q_i - q_i^{(1)})
=
\frac{a_\ell}{b_\ell}(1 - 1)
= 0.
\]
So $L(\mbf q) \ge L(\mbf q^{(1)})$ for every $\mbf q \in U$, and hence
$\mbf q^{(1)}$ minimizes $L(\mbf q)$ over $U$.

We may compute
\begin{align}
\min_{\mbf q \in U} L(\mbf q) = \sum_{i=1}^{\ell-1} \frac{a_i}{b_i}  \frac{y_i}{a_i}
+ \frac{a_\ell}{b_\ell} \left(1 - \sum_{i=1}^{\ell-1} \frac{y_i}{a_i} \right).
\label{eq:L(q)-lowbound}
\end{align}
By a similar argument, 
\begin{align}
\max_{\mbf q \in U} L(\mbf q) = \sum_{i=\ell'+1}^{n} \frac{a_i}{b_i}  \frac{y_i}{a_i}
+ \frac{a_{\ell'}}{b_{\ell'}} \left(1 - \sum_{i=\ell'+1}^{n} \frac{y_i}{a_i} \right).
\label{eq:L(q)-upbound}
\end{align}
for some $\ell' \in V$, where
\[
\sum_{i=\ell'+1}^{n} \frac{y_i}{a_i} < 1 \le \sum_{i=\ell'}^{n} \frac{y_i}{a_i}.
\]
Thus, $L(U)$ is the interval determined by~(\ref{eq:L(q)-lowbound}) and~(\ref{eq:L(q)-upbound}), since $U$ is convex and compact and $L$ is a linear map. To complete the proof, it is sufficient to show that $T$ lies in this interval. Rearranging~(\ref{eq:ineq1}), with $m=\ell$, we obtain
\[
T
\ge
\sum_{i=1}^{\ell-1} \frac{a_i}{b_i}  \frac{y_i}{a_i}
+ \frac{a_\ell}{b_\ell} \left(1 - \sum_{i=1}^{\ell-1} \frac{y_i}{a_i} \right).
\]
The right-hand side is the minimum of $L(\mbf q)$, as computed in~(\ref{eq:L(q)-lowbound}).

Similarly, rearranging~(\ref{eq:ineq2}, with $m=\ell$, we obtain
\[
a_{\ell'} \sum_{i={\ell'}+1}^{n} \frac{y_i}{a_i}
+
b_{\ell'} \sum_{i=1}^{\ell'} \frac{y_i}{b_i}
\le
a_{\ell'} + b_{\ell'},
\]
and this time, the right-hand side is the maximum of $L(\mbf q)$, as in~(\ref{eq:L(q)-upbound}). Thus, $T \in L(U)$, and the proof is complete.
\end{proof}

\begin{remark} It is easy to see from Lemma~\ref{lem:ineq} that there must be an optimal strategy for Player 2 for which one of the inequalities~(\ref{eq:ineq1}) hold with equality. 
Indeed, suppose $\mbf y$ is an optimal mixed strategy that does not satisfy any of inequalities~(\ref{eq:ineq1}) with equality, and let $\varepsilon_m >0 $ be the difference between the left-hand side and the right-hand side of inequality~(\ref{eq:ineq1}), for each $m \in V$. Let $\delta = \min_{m \in V} \varepsilon_m a_1/a_m$. Now let $\mbf y'$ be the same as $\mbf y$, except for the first coordinate, which is given by $y'_1=y_1 - \delta$. 
Then replacing $\mbf y$ with $\mbf y'$, the left-hand side of inequality~(\ref{eq:ineq1})  decreases by $\delta a_m/a_1 \le \varepsilon_m$ for each $m$, and this inequality holds with equality for some $m$. 
Hence, each of the inequalities~(\ref{eq:ineq1}) hold, and at least one holds with equality when $\mbf y$ is replaced with $\mbf y'$.
Lastly, it is clear that all the inequalities~(\ref{eq:ineq2}) continue to hold.
So $\mbf y'$ is also a mixed strategy, and it clearly (weakly) dominates $\mbf y$. Thus, it must be optimal.
\end{remark}

We now give sufficient conditions to be able to find closed form optimal strategies for each player. This will allow us to find optimal strategies in the quiz show game when $M$ is large compared to~$C$ -- a reasonable assumption for the case of information gathering drones, where the main prizes and consolation prizes correspond to the information that could be gathered at a site with or without capture, respectively.

\begin{theorem} \label{thm:main}
    Let $k= \min \{\ell:\sum_{i=l+1}^n 1/b_i \le \sum_{i=1}^l 1/a_i\}$ and 
    \[
    \mbf y^k = \lambda_k \mbf 1,
    \]
    with $\mbf 1$ representing the vector of ones and
    \[
    \lambda_k = \frac{a_k+b_k}{a_k \sum_{i=1}^k 1/a_i + b_k \sum_{i=k+1}^n 1/b_i}.
    \]
    Then $\mbf y^k$ is a mixed strategy for Player 2, and 
    \[ P(\mbf x, \mbf y^k) \le v_k \equiv \lambda_k f(V).\] Furthermore, the strategy $\mbf y^k$ is optimal if $\mbf x^k$ lies in $\mathcal{B}(f)$, where
    \[
    (\mbf x^k)^T = \frac{\lambda_k f(V)}{a_k+b_k} (a_k/a_1,\ldots,a_k/a_k,b_k/b_{k+1},\ldots,b_k/b_n).
    \]
    In this case, $\mbf x^k$ is also optimal, and the value of the game is $v_k$.
\end{theorem}
\begin{proof}
We first show that $\mbf y^k$ is a mixed strategy for Player 2.
Define $\mbf q, \mbf q'$ by
\[
q_i =
\begin{cases}
\dfrac{\lambda_k}{a_i}, & i < k, \\[6pt]
1 - \lambda_k \displaystyle\sum_{i=1}^{k-1} \frac{1}{a_i}, & i = k, \\[8pt]
0, & i > k,
\end{cases}
\]
and
\[
q'_i =
\begin{cases}
0, & i < k, \\[6pt]
1 - \lambda_k \displaystyle\sum_{i=k+1}^{n} \frac{1}{b_i}, & i = k, \\[8pt]
\dfrac{\lambda_k}{b_i}, & i > k.
\end{cases}
\]

We will show that $\mbf q$ and $\mbf q'$ are probability vectors. First,
\[
\sum_{i=1}^{n} q_i
=
\lambda_k \sum_{i=1}^{k-1} \frac{1}{a_i}
+
\left(1 - \lambda_k \sum_{i=1}^{k-1} \frac{1}{a_i}\right)
= 1,
\]
and similarly,
\[
\sum_{i=1}^{n} q_i'
=
\left(1 - \lambda_k \sum_{i=k+1}^{n} \frac{1}{b_i}\right)
+
\lambda_k \sum_{i=k+1}^{n} \frac{1}{b_i}
= 1.
\]
It remains to show nonnegativity of $q_k$ and $q'_k$. Using the definition of $\lambda_k$, the condition $q_k \ge 0$ is equivalent to
%\[
%lambda_k \sum_{i=1}^{k-1} \frac{1}{a_i} \le 1,
%\]
%which is equivalent (by the definition of $\lambda_k$) to
%\[
%(a_k + b_k) \sum_{i=1}^{k-1} \frac{1}{a_i}
%\le
%1
%+
%a_k \sum_{i=1}^{k-1} \frac{1}{a_i}
%+
%b_k \sum_{i=k+1}^{n} \frac{1}{b_i},
%\]

%\[
%b_k \sum_{i=1}^{k-1} \frac{1}{a_i}
%\le
%1 + b_k \sum_{i=k+1}^{n} \frac{1}{b_i}.
%\]

%This is equivalent to
\[
\sum_{i=1}^{k-1} \frac{1}{a_i}
\le
%\frac{1}{b_k}
%+
%\sum_{i=k+1}^{n} \frac{1}{b_i}
%=
\sum_{i=k}^{n} \frac{1}{b_i},
\]
which follows from the minimality of $k$.

Similarly, the condition $q_k' \ge 0$ is equivalent to
%\[
%\lambda_k \sum_{i=k+1}^{n} \frac{1}{b_i} \le 1,
%\]
%which is equivalent to
%\[
%(a_k + b_k) \sum_{i=k+1}^{n} \frac{1}{b_i}
%\le
%1
%+
%a_k \sum_{i=1}^{k-1} \frac{1}{a_i}
%+
%b_k \sum_{i=k+1}^{n} \frac{1}{b_i},
%\]
%\[
%a_k \sum_{i=k+1}^{n} \frac{1}{b_i}
%\le
%1 + a_k \sum_{i=1}^{k-1} \frac{1}{a_i},
%\]
%which is equivalent to
\[
\sum_{i=k+1}^{n} \frac{1}{b_i}
\le
%\frac{1}{a_k}
%+
%\sum_{i=1}^{k-1} \frac{1}{a_i}
%=
\sum_{i=1}^{k} \frac{1}{a_i},
\]
again true by the minimality of $k$.

Thus $\mbf q, \mbf q' \ge 0$, and both are probability vectors. We now show that $ y^k_i = \lambda_k$ is equal to $q_i a_i + q_i' b_i$, so that $\mbf y^k \in \Delta(\mbf a) + \Delta(\mbf b)$. This is trivially true for $i<k$ or $i>k$. For $i = k$,
\begin{align}
a_k q_k + b_k q_k'
&=
a_k\!\left(1 - \lambda_k \sum_{i=1}^{k-1} \frac{1}{a_i}\right)
+
b_k\!\left(1 - \lambda_k \sum_{i=k+1}^{n} \frac{1}{b_i}\right)
\\
&=
a_k + b_k
-
\lambda_k\!\left(
a_k \sum_{i=1}^{k-1} \frac{1}{a_i}
+
b_k \sum_{i=k+1}^{n} \frac{1}{b_i}
\right). \\
&= \lambda_k,
\end{align}
by definition of $\lambda_k$. Therefore, $\mbf y^k \in \Delta(\mbf a) + \Delta(\mbf b)$, so $\mbf y^k$ is a Player~2 mixed strategy. 

To show that $\mbf y$ guarantees an expected payoff of at most $v_k$ against any Player~1 strategy $\mbf x \in \mathcal{B}(f)$, we calculate the payoff
\[
P(\mbf x,\mbf y^k) = \mbf x^T \mbf y^k
= \lambda_k \sum_{i=1}^{n} x_i
= \lambda_k f(V)
= v_k ,
\]
where the penultimate equality follows from $\mbf x \in \mathcal{B}(f)$.

Assume now that $\mbf x^k \in \mathcal{B}(f)$, and we will show that $\mbf x^k$ guarantees a payoff at least $v_k$ against every Player 2 mixed strategy $\mbf y \in \Delta(\mbf a) + \Delta(\mbf b)$. Indeed, this payoff is
\begin{align*}
(\mbf x^k)^T \mbf y
&=
\frac{\lambda_k f(V)}{a_k + b_k}
\left(
a_k \sum_{i=1}^{k} \frac{y_i}{a_i}
+
b_k \sum_{i=k+1}^{n} \frac{y_i}{b_i}
\right) \\
&\ge
\lambda_k f(V) = v_k%\\
%&=
%v_k,
\end{align*}
where the inequality follows from~(\ref{eq:ineq1}), with $m=k$.

Thus, $\mbf x^k$ guarantees payoff at least $v_k$.

Therefore
\[
\sup_{\mbf x \in B(f)} P(\mbf x,\mbf y^k)
\le
v_k
\le
\inf_{\mbf y \in \Delta(\mbf a)+\Delta(\mbf b)} P(\mbf x^k,\mbf y).
\]
Hence, both inequalities must hold with equality, and $(\mbf x^k,\mbf y^k)$ is an equilibrium pair. In particular, both strategies are optimal and the value of the game is $v_k$.
\end{proof}

Note that $\mbf x^k$, as defined in Theorem~\ref{thm:main} always satisfies $\mbf x^k(V)=f(V)$, so $\mbf x^k$ lies in $\mathcal{B}(f)$ if and only if $\mbf x^k(S) \le f(S)$ for all $S \subsetneq V$.

The theorem gives sufficient conditions for Player~2 to be able to make Player~1 indifferent between all her strategies in an equilibrium. This is because the indices $y^k_i$ are all equal. 

It is also worth pointing out that the strategy $\mbf y^k$ makes inequality~(\ref{eq:ineq1}) hold with equality for $m=k$. One might conjecture that the optimal solution is always on the corresponding facet of Player 2's mixed strategy space, whether or not $\mbf y^k$ is optimal, but we will show in Subsection~\ref{sec:numerical} that this is not necessarily the case.

We now use Theorem~\ref{thm:main} to solve this variant of the quiz show game in the case that $M$ is large compared to $C$.

\begin{corollary}
\label{cor:big}    In the quiz show game, if $M/C > \sum_{i=1}^{n-1} \frac{1-\alpha_i}{\alpha_i}$, then the strategies $\mbf x^n$ and $\mbf y^n$ are optimal, and the value of the game is 
    \[
    v_n = \frac{(M\alpha_n/(1-\alpha_n) + C)(1-\prod_{i=1}^n \alpha_i)}{\alpha_n/(1-\alpha_n) \sum_{i=1}^n (1-\alpha_i)/\alpha_i}.
    \]
\end{corollary}
\begin{proof}
Recall that in the quiz show game, the parameters $a_i$ and $b_i$ are given by
\[
a_i = \frac{M \alpha_i}{1 - \alpha_i},
\qquad
b_i = C,
\qquad
i = 1,\ldots,n.
\]
In order to use Theorem \ref{thm:main}, we first show that the minimal $k$ to satisfy
\begin{align}
\sum_{i=k+1}^{n} \frac{1}{b_i}
\le
\sum_{i=1}^{k} \frac{1}{a_i} \label{eq:cor}
\end{align}
is $k=n$. Indeed, when $k = n$, inequality~(\ref{eq:cor}) trivially holds. Moreover, for any $k < n$, we have
\[
    \sum_{i=k+1}^n \frac{1}{b_i} = \frac{n-k}{C} > \frac{n-k}{M} \sum_{i=1}^{n-1} \frac{1-\alpha_i}{\alpha_i} 
     \ge \frac{1}{M} \sum_{i=1}^{k} \frac{1-\alpha_i}{\alpha_i} 
     = \sum_{i=1}^k \frac{1}{a_i},
\]
where the first inequality follows from the assumption of the corollary. Consequently $k = n$ is the minimal $k$ such that~(\ref{eq:cor}) holds.

%Thus some admissible $k$ always exists. The question is whether a smaller one could already satisfy the inequality.

%Now consider $k = n-1$. The defining inequality becomes
%\[
%\frac{1}{b_n} \le \sum_{i=1}^{n-1} \frac{1}{a_i}.
%\]

%Substituting $b_n = C$ and $a_i = \dfrac{M \alpha_i}{1 - \alpha_i}$,  we get
%\[
%\frac{M}{C} \le \sum_{i=1}^{n-1} \frac{1 - \alpha_i}{\alpha_i}.
%\]

%Therefore, if
%\[
%\frac{M}{C} > \sum_{i=1}^{n-1} \frac{1 - \alpha_i}{\alpha_i},
%\]
%then the inequality for $k = n-1$ fails and hence no index $k < n$ can be minimal.

The corollary follows from Theorem \ref{thm:main} as long as $\mbf x^n$ is a feasible strategy for Player 1. We will show that $\mbf x^n$ is indeed feasible by proving that it lies in $\mathcal{B}(f)$. As already noted, it is sufficient to show that for any $S \subsetneq V$, we have $\mbf x^n(S) \le f(S)$. To do this, we first define
\[
\phi(S)
:=
\frac{1 - \prod_{i \in S} \alpha_i}
{\sum_{i \in S} \frac{1-\alpha_i}{\alpha_i}}.
\]

We will show that $\phi$ is non-increasing with respect to inclusion. This is sufficient to prove the corollary, as $\phi(V) \le \phi(S)$ is equivalent to $\mbf x^n(S) \le f(S)$.

If $T = S \cup \{m\}$ with $m \notin S$, then $\phi(T) \le \phi(S)$ is equivalent, after cross-multiplying, to
\[
\left(1 - \alpha_m \prod_{i \in S} \alpha_i\right)
\sum_{i \in S} \frac{1 - \alpha_i}{\alpha_i}
\le
\left(1 - \prod_{i \in S} \alpha_i\right)
\left(
\frac{1 - \alpha_m}{\alpha_m}
+
\sum_{i \in S} \frac{1 - \alpha_i}{\alpha_i}
\right),
\]
that is,
\[
\left(\prod_{i \in S} \alpha_i\right)
(1 - \alpha_m)
\sum_{i \in S} \frac{1 - \alpha_i}{\alpha_i}
\le
\left(1 - \prod_{i \in S} \alpha_i\right)
\frac{1 - \alpha_m}{\alpha_m}.
\]

If $1 - \alpha_m = 0$ there is nothing to prove; otherwise dividing by
$1 - \alpha_m > 0$ and rewriting, it is enough to show
\[
\sum_{i \in S}\left(\prod_{j \in S, j \neq i} \alpha_j\right)(1 - \alpha_i)
\le
\frac{1 - \prod_{i \in S} \alpha_i}{\alpha_m},
\]
and since $\alpha_m \le 1$, it suffices to prove
\[
\sum_{i \in S}\left(\prod_{j \in S, j \neq i} \alpha_j\right)
(1 - \alpha_i)
\le
1 - \prod_{i \in S} \alpha_i,
\]
Using the identity $\prod_{i \in S} (x_i + y_i) = \sum_{I \subseteq S} \prod_{i \in I} x_i \prod_{i \notin I} y_i$ with $x_i = 1-\alpha_i$ and $y_i = \alpha_i$, we obtain
\[
1-\prod_{i \in S} \alpha_i = \sum_{\emptyset \neq I \subseteq S}
\prod_{i \in I} (1 - \alpha_i)
\prod_{j \in S \setminus I} \alpha_j.
\]
Taking $I$ to be the singletons, it is clear that the sum on the right-hand side contains
\[
\sum_{i \in S} (1 - \alpha_i)
\prod_{j \in S,\, j \ne i} \alpha_j
\]
plus additional nonnegative terms. Therefore $\phi(T) \le \phi(S)$,
so that $\phi$ is non-increasing under inclusion.
Hence, $\phi(V) \le \phi(S)$, which completes the proof.
\end{proof}

Theorem~\ref{thm:main} is only useful if we are able to determine whether a given $\mbf x^k$ lies in $\mathcal{B}(f)$: that is, whether $x^k(S) \le f(S)$ for all $S \subseteq V$. Equivalently, we wish to show that $\min_{S \subseteq V} (f(S)-\mbf x^k(S)) \ge 0$. the function $S \mapsto f(S) - \mbf x^k(S)$ is submodular, so it follows from standard results on minimizing submodular functions that there is a strongly polynomial time algorithm to determine whether $\mbf x^k$ lies in $\mathcal{B}(f)$.

In the case of the quiz show game, we now show that determining whether $\mbf x^k$ lies in $\mathcal{B}(f)$ can be done even quicker, simply by checking the inequality $\mbf x^k(S) \le f(S)$ for $n$ different sets $S$.

\begin{lemma}
\label{lem:algor}    Let $f$ be the submodular function of the quiz show game, given by $f(S) = 1- \prod_{i \in S} \alpha_i$, for fixed $\alpha_1,\ldots,\alpha_n$, and let $\mbf x \in \mathbb{R}^n$ satisfy $\mbf x(V)=f(V)$. Let $\sigma$ be some permutation such that $x_{\sigma(1)}/(1-\alpha_{\sigma(1)}) \ge \cdots \ge x_{\sigma(n)}/(1-\alpha_{\sigma(n)})$, and let $S_j := \{\sigma(1),\ldots,\sigma(j)\}$ for each $j \in V$. Then
    \[
\min_{S \subseteq V} (f(S)-\mbf x(S)) = \min_{j \in V} (f(S_j) - \mbf x(S_j)).
    \]
\end{lemma}
\begin{proof}
    We prove the claim by contradiction. Suppose $S^* \subseteq V$ minimizes $f(S)-\mbf x(S)$, and
\[
f(S^*) - \mbf x(S^*) < \min_{j \in V} \bigl(f(S_j) - \mbf{x}(S_j)\bigr).
\]
Let $m := \max \{i:\sigma(i) \in S^*\}$. Then $\sigma(m) \in S^*$ and $S^* \subseteq S_m$. Define $A := S_m \setminus S^*$.

By the submodularity of $f$,
\begin{align}
f(S_m) - f(S^*)
=
f(S^* \cup A) - f(S^*)
\le
\sum_{i \in A} \bigl(f(S^* \cup \{i\}) - f(S^*)\bigr).\label{eq:submod}
\end{align}

Since $S^*$ minimizes $f(S)-\mbf x(S)$,
\[
f(S^*) - \mbf x(S^*) \le f(S^* \setminus \{\sigma(m)\}) - \mbf x(S^* \setminus \{\sigma(m)\}).
\]
Thus,
\begin{align}
f(S^*) - f(S^* \setminus \{\sigma(m)\}) \le x_{\sigma(m)}.
\label{eq:S-le}    
\end{align}
Using the definition of $f$ in~(\ref{eq:S-le}) and rearranging gives
\begin{align*}
\prod_{i \in S \setminus \{{\sigma(m)}\}} \alpha_i
&\le
\frac{x_{\sigma(m)}}{1 - \alpha_{\sigma(m)}} \\
& \le \frac{x_{\sigma(k)}}{1 - \alpha_{\sigma(k)}},
\end{align*}
for any other $\sigma(k) \in A$, since the sequence
$\frac{x_{\sigma(i)}}{1 - \alpha_{\sigma(i)}}$ is non-increasing in $i$. Therefore,
\begin{align}
  (1 - \alpha_{\sigma(j)})\prod_{i \in S \setminus \{\sigma(m)\}} \alpha_i \le x_{\sigma(j)}. \label{eq:x_i-le}  
\end{align}\
On the other hand,
\[
f(S^* \cup \{\sigma(j)\}) - f(S^*)
=
(1 - \alpha_{\sigma(j)})\prod_{i \in S^*} \alpha_i
\le
(1 - \alpha_{\sigma(j)})\prod_{i \in S^* \setminus \{\sigma(m)\}} \alpha_i,
\]
so combining with~(\ref{eq:x_i-le}),
\[
f(S^* \cup \{\sigma(j)\}) - f(S^*) \le x_{\sigma(j)}.
\]
Summing over $i=\sigma(j) \in A$,
\[
\sum_{\i \in A} \bigl(f(S^* \cup \{i\}) - f(S^*)\bigr)
\le \mbf x(A).
\]

Together with~(\ref{eq:submod}), this yields
\[
f(S_m) - f(S^*) \le \mbf x(A).
\]

Hence,
\[
f(S_m) - \mbf x(S_m)
=
f(S_m) -\mbf  x(S^*) - \mbf x(A)
\le
f(S^*) - \mbf x(S^*).
\]
It follows that $S_m$ minimizes $f(S)-\mbf x(S)$, which completes the proof.
\end{proof}
\subsection{Numerical Results} \label{sec:numerical}
Below we present several examples for the quiz show game that illustrate the role
of Theorem~\ref{thm:main}, Corollary~\ref{cor:big}, and Lemma~\ref{lem:algor}. Recall from Theorem~\ref{thm:main} that, when the
host plays $\mbf{y^k}$, this strategy provides the upper bound $v_k$ on the value of
the game. However, the strategy $\mbf{x^k}$ defined in Theorem~\ref{thm:main} may not lie in
$\mathcal B(f)$. Consequently, the value of the game may be strictly smaller
than~$v_k$.

Throughout this section, without loss of generality, we set $C=1$. Therefore $M/C$ is $M$. We consider two families of examples,
each with a different choice of the vector $\mbf \alpha$. In the first family, we take
$\mbf{\alpha}^T=(0.2,0.4,0.9)$. For various values of $M$, Table~\ref{tab:sol1} reports the parameters defined in
Theorem~\ref{thm:main}, including $k$, $\lambda_k$, $v_k$, and the actual value of the game,
denoted by $v$. In order to calculate $v$ we used Game Theory Explorer \citep{savani2015gte}, available at http://app.test.logos.bg/. The lower bound on $M/C=M$ given in Corollary~\ref{cor:big} is equal to $5.5$.
Thus, when $M>5.5$, we have $k=3$, and the value of the game is equal to
$v_3$, as predicted by Corollary~\ref{cor:big}. However, when $M\le 5.5$, the strategy
$\mbf x^k$ does not lie in $\mathcal B(f)$, as can be checked using Lemma~\ref{lem:algor}. Hence,
$v_k$ is only an upper bound and is not, in general, the value of the game.

%Here, $v$ and $v_k$ varies when $k=1,2$. $M/C=2$ is the boundary of $k=1,2$; $M/C=5.5$ is the boundary of $k=2,3$. The intro of the head of the chart. how to find $v_k$ and $v$. the fact of decresing
%lemma 6 is used when $k=2$

%\begin{table}
%\centering
%\caption{example 1}
%\label{tab:dist_by_condition}
%\footnotesize
%\setlength{\tabcolsep}{3pt}
%\resizebox{\columnwidth}{!}{%
\begin{table}[h]
    \centering
    \caption{Comparison of $v_k$ with value of game for $\mbf \alpha^T=(0.2,0.4,0.9)$}
    \label{tab:sol1}
    \begin{tabular}{c|c|c|c|c|c}
        $M$ & $k$ & $\lambda_k$ & $v_k$ & $v$ & $v/v_k$ \\
        \hline
0.001 & 1 & 0.3334 & 0.3094 & 0.1002 & 32.39\%  \\
0.1   & 1 & 0.3417 & 0.3171 & 0.1151 & 36.30\% \\
1     & 1 & 0.4167 & 0.3867 & 0.2505 & 64.78\%  \\
1.9   & 1 & 0.4917 & 0.4563 & 0.3860 & 84.59\%  \\
2.1   & 2 & 0.5143 & 0.4773 & 0.4161 & 87.18\% \\
5.4   & 2 & 0.9857 & 0.9147 & 0.9129 & 99.80\% \\
\hline
Corollary 5 Lower Bound $=5.5$ \\
\hline
5.6   & 3 & 1.0178 & 0.9445 & 0.9445 & 100.00\% \\
10    & 3 & 1.8200 & 1.6722 & 1.6722 & 100.00\% \\
100   & 3 & 17.8416 & 16.5570 & 16.5570 & 100.00\%
    \end{tabular}
\end{table}

%For instance of that $\alpha$'s are 0.2 0.8 0.9 with various $M/C$, it shows the different $v_k$ and $v$. Here, $v$ and $v_k$ varies when $k=1$. $M/C=2$ is the boundary of $k=1,2$; $M/C=4.25$ is the boundary of $k=2,3$.

\begin{table}[h]
    \centering
    \caption{Comparison of $v_k$ with value of game for $\mbf \alpha^T=(0.2,0.8,0.9)$}
    \label{tab:sol2}
    \begin{tabular}{c|c|c|c|c|c}
$M$ & $k$ & $\lambda_k$ & $v_k$ & $v$ & $v/v_k$ \\
\hline
0.1 & 1 & 0.3417 & 0.2925 & 0.1164 & 39.79\% \\
1 & 1 & 0.4167 & 0.3567 & 0.2640 & 74.01\% \\
1.9 & 1 & 0.4917 & 0.4209 & 0.4116 & 97.79\% \\
2.1 & 2 & 0.5222 & 0.4470 & 0.4470 & 100.00\% \\
4.2 & 2 & 0.9889 & 0.8465 & 0.8465 & 100.00\% \\
\hline
\mbox{Corollary 5 Lower Bound=4.25} \\
\hline4.3 & 3 & 1.0115 & 0.8658 & 0.8658 & 100.00\% \\
10 & 3 & 2.3185 & 1.9846 & 1.9846 & 100.00\%
    \end{tabular}
\end{table}

In the second family of examples, we take $\mbf{\alpha}^T=(0.2,0.8,0.9)$. The results are shown in Table~\ref{tab:sol2}.
Here, although $M<4.25$ implies that $k=2$, Lemma~\ref{lem:algor} can be used to show that for $M=2.1$ and $M=4.2$,
$\mbf x^2$ lies in $\mathcal B(f)$. Therefore, $\mbf{x^2}$ is feasible and optimal, and
$v_2$ is the value of the game.

The last column of the tables report the ratio $v/v_k$, which measures how tight the upper
bound $v_k$ is. A ratio of $100\%$ indicates that the upper bound is attained and
therefore $v=v_k$. Ratios strictly below $100\%$ indicate precisely the cases in
which the strategy $\mbf{x^k}$ from Theorem~\ref{thm:main} is not feasible for the quiz show game,
so that the actual game value is strictly smaller than the bound supplied by
$\mbf y^k$. It is interesting to note that for our numerical examples, the ratio $v/v_k$ is non-decreasing in $M$. We conjecture that this is true in general.

Recall from the discussion following Lemma~\ref{lem:ineq} that there must be an optimal strategy for Player~2 for which one of the inequalities~(\ref{eq:ineq1}) holds with equality. We call the set of feasible strategies for which~(\ref{eq:ineq1}) holds with equality (for a given $m$) the {\em $m$-type facet} of Player 2's strategy set. In the first family of examples, when $M>5.5$ ($k=3$), the optimal $\mbf y^*$ lies on the 3-type facet (as we know from Corollary~\ref{cor:big}); however, when $M<5.5$, whether $k$ is $1$ or $2$, the optimal $\mbf y^*$ lies on the intersection of the 2-type facet and the 3-type facet. 
%In the second family of examples, when $M>4.25$ ($k=3$), the optimal $ \mbf y^*$ lies on 3-type facet; when $4.25>M/C>2$ ($k=2$), the optimal $y^*$ lies on 2-type facet; when $2>M/C$ ($k=1$), the optimal $y^*$ lies on the intersection of 1-type facet and 2-type facet. Note that the equilibrium is obtained by the site...
In the second family of examples, except in the case of $k=1$, whatever the value of $k$ is, the optimal solution found lies on the $k$-type facet.

\section{Conclusion}

We have introduced three variants of a quiz show game, motivated by national security applications where information may be unveiled in two different ways. Finding solutions to two of the variants was fairly straightforward, but a general solution to the third variant remains elusive. However, we were able to prove some elegant structural properties of this variant by regarding it as a special case of a more general game, and we are optimistic that further work may yield a more complete characterization of its optimal strategies.

\section*{Acknowledgments}
This material is based upon work supported by the Air Force Office of
Scientific Research under award number FA9550-23-1-0556.

\bibliographystyle{apalike}
\bibliography{references}

\begin{thebibliography}{}

\bibitem[Agnetis et~al., 2022]{agnetis2022replication}
Agnetis, A., Benini, M., Detti, P., Hermans, B., and Pranzo, M. (2022).
\newblock Replication and sequencing of unreliable jobs on parallel machines.
\newblock {\em Computers \& Operations Research}, 139:105634.

\bibitem[Agnetis et~al., 2025a]{agnetis2025replication}
Agnetis, A., Benini, M., Detti, P., Hermans, B., Pranzo, M., and Schewior, K.
  (2025a).
\newblock Replication and sequencing of unreliable jobs on m parallel machines:
  New results.
\newblock {\em Computers \& Operations Research}, 183:107085.

\bibitem[Agnetis et~al., 2009]{agnetis2009sequencing}
Agnetis, A., Detti, P., Pranzo, M., and Sodhi, M.~S. (2009).
\newblock Sequencing unreliable jobs on parallel machines.
\newblock {\em Journal of Scheduling}, 12:45--54.

\bibitem[Agnetis et~al., 2025b]{agnetis2025unreliable}
Agnetis, A., Leus, R., Perneel, E., and Salvadori, I. (2025b).
\newblock The unreliable job selection and sequencing problem.
\newblock {\em arXiv preprint arXiv:2511.17105}.

\bibitem[Agnetis and Lidbetter, 2020]{agnetis2020largest}
Agnetis, A. and Lidbetter, T. (2020).
\newblock The largest-z-ratio-first algorithm is 0.8531-approximate for
  scheduling unreliable jobs on m parallel machines.
\newblock {\em Operations Research Letters}, 48(4):405--409.

\bibitem[Alpern and Gal, 2003]{alpern&gal03book}
Alpern, S. and Gal, S. (2003).
\newblock {\em {The Theory of Search Games and Rendezvous}}.
\newblock Springer.

\bibitem[Bram, 1963]{bram19632}
Bram, J. (1963).
\newblock A 2-player n-region search game.
\newblock {\em OEG IRM-31 (AD 402914), Washington}.

\bibitem[Fonlupt and Skoda, 2009]{fonlupt2009strongly}
Fonlupt, J. and Skoda, A. (2009).
\newblock Strongly polynomial algorithm for the intersection of a line with a
  polymatroid.
\newblock In {\em Research Trends in Combinatorial Optimization: Bonn 2008},
  pages 69--85. Springer.

\bibitem[Fujishige, 1980]{fujishige1980lexicographically}
Fujishige, S. (1980).
\newblock Lexicographically optimal base of a polymatroid with respect to a
  weight vector.
\newblock {\em Mathematics of Operations Research}, 5(2):186--196.

\bibitem[Garnaev, 2012]{garnaev2012search}
Garnaev, A. (2012).
\newblock {\em Search games and other applications of game theory}, volume 485.
\newblock Springer Science \& Business Media.

\bibitem[Gr{\"o}tschel et~al., 2012]{grotschel2012geometric}
Gr{\"o}tschel, M., Lov{\'a}sz, L., and Schrijver, A. (2012).
\newblock {\em Geometric algorithms and combinatorial optimization}, volume~2.
\newblock Springer Science \& Business Media.

\bibitem[Hellerstein and Lidbetter, 2023]{hellerstein2023game}
Hellerstein, L. and Lidbetter, T. (2023).
\newblock A game theoretic approach to a problem in polymatroid maximization.
\newblock {\em European Journal of Operational Research}, 305(2):979--988.

\bibitem[Hoeksma et~al., 2014]{hoeksma2014decomposition}
Hoeksma, R., Manthey, B., and Uetz, M. (2014).
\newblock Decomposition algorithm for the single machine scheduling polytope.
\newblock In {\em International Symposium on Combinatorial Optimization}, pages
  280--291. Springer.

\bibitem[Hohzaki, 2016]{hohzaki2016search}
Hohzaki, R. (2016).
\newblock Search games: Literature and survey.
\newblock {\em Journal of the Operations Research Society of Japan},
  59(1):1--34.

\bibitem[Isaacs, 1965]{Isaacs-Book-1965}
Isaacs, R. (1965).
\newblock {\em Differential Games}.
\newblock Wiley, New York.

\bibitem[Kadane, 1969]{kadane1969quiz}
Kadane, J.~B. (1969).
\newblock Quiz show problems.
\newblock {\em Journal of Mathematical Analysis and Applications}, 26:609--623.

\bibitem[Lidbetter, 2020]{lidbetter2020search}
Lidbetter, T. (2020).
\newblock Search and rescue in the face of uncertain threats.
\newblock {\em European Journal of Operational Research}, 285(3):1153--1160.

\bibitem[Lidbetter, 2025]{lidbetter2025review}
Lidbetter, T. (2025).
\newblock A review of minimum cost box searching games.
\newblock {\em arXiv preprint arXiv:2502.10551}.

\bibitem[Savani and von Stengel, 2015]{savani2015gte}
Savani, R. and von Stengel, B. (2015).
\newblock Game theory explorer: Software for the applied game theorist.
\newblock {\em Computational Management Science}, 12(1):5--33.

\bibitem[Stadje, 1995]{stadje1995selecting}
Stadje, W. (1995).
\newblock Selecting jobs for scheduling on a machine subject to failure.
\newblock {\em Discrete applied mathematics}, 63(3):257--265.

\bibitem[Stone, 1976]{stone1976theory}
Stone, L.~D. (1976).
\newblock {\em Theory of optimal search}, volume 118.
\newblock Elsevier.

\bibitem[Yolmeh and Baykal-G{\"u}rsoy, 2021]{yolmeh2021weighted}
Yolmeh, A. and Baykal-G{\"u}rsoy, M. (2021).
\newblock Weighted network search games with multiple hidden objects and
  multiple search teams.
\newblock {\em European Journal of Operational Research}, 289(1):338--349.

\end{thebibliography}

\end{document}